%% file: Paper-Allg.tex
\documentclass[11pt]{article}

\newcommand{\f}{\frac}

\newcommand{\del}{\partial}
\newcommand{\la}{\label}

\newcommand{\R}{\mathbb R}
\newcommand{\LR}{L^2(\R)}
\newcommand{\re}{\operatorname{Re}}
\newcommand{\imag}{\operatorname{Im}}
\newcommand{\eps}{\varepsilon}

\newcommand{\ran}{\operatorname{Ran}}

\usepackage{etex}
\usepackage[english]{babel}
\usepackage{tikz}
\usepackage{pgfplots}
\usetikzlibrary{fadings,fpu,arrows,decorations.pathmorphing,backgrounds,positioning,fit,petri,patterns}
\usepackage[utf8]{inputenc}
\usepackage{graphicx}
\usepackage[colorlinks=true,linkcolor=black,citecolor=black]{hyperref}
\usepackage{amsfonts}
\usepackage{enumerate}
\usepackage{booktabs}
\usepackage{array} 
\usepackage{paralist} 
\usepackage{subfig} 
\usepackage{mathtools}
\usepackage{tabu}
\usepackage{amsthm}
\usepackage{empheq}
\usepackage{amsopn}
\usepackage{dsfont}
\usepackage{wrapfig}
\usepackage[font=small]{caption}

\DeclareMathOperator{\Ran}{{Ran}}
\DeclareMathSizes{10.95}{10}{7}{7} 

\theoremstyle{definition}
\newtheorem{de}{Definition}[section]

\theoremstyle{plain}
\newtheorem{prop}[de]{Proposition}
\newtheorem{lemma}[de]{Lemma}
\newtheorem{theorem}[de]{Theorem}
\newtheorem{cor}[de]{Corollary}

\numberwithin{equation}{section}

\author{Patrick Dondl\thanks{Abteilung f\"ur angewandte Mathematik\newline 
Albert-Ludwigs-Universit\"at Freiburg \newline
Hermann-Herder-Str. 10\newline
79104 Freiburg i. Br \newline
Germany \newline
Email: \{patrick.dondl, frank.roesler\}@math.uni-freiburg.de},
Patrick Dorey\thanks{Department of Mathematical Sciences\newline 
Durham University \newline
South Road \newline
Durham DH1 3LE \newline
United Kingdom\newline
Email: p.e.dorey@durham.ac.uk},
Frank R\"{o}sler\footnotemark[1]
}

\title{A Bound on the Pseudospectrum for a Class of Non-Normal Schr\"odinger Operators}

\begin{document}

\maketitle

\begin{abstract}
\noindent We are concerned with the non-normal Schr\"odinger operator
 $	
 H=-\Delta+V
 $
on $ L^2(\mathbb R^n)$, where $V\in W^{1,\infty}_{\text{loc}}(\R^n)$ and $\re V(x)\ge c|x|^2-d$ for some $c,d>0$. The spectrum of this operator is discrete and contained in the positive half plane. In general, the $\eps$-pseudospectrum of $H$ will have an unbounded component for any $\eps>0$ and thus will not approximate the spectrum in a global sense.

By exploiting the fact that the semigroup $e^{-tH}$ is immediately compact, we  show a complementary result, namely that  for every $\delta>0$, $R>0$ there exists an $\eps>0$ such that the $\eps$-pseudospectrum
\begin{equation*}
 	\sigma_\eps(H)\subset \{z:\re z \geq R\}\cup\bigcup_{\lambda\in\sigma(H)}\{z:|z-\lambda|<\delta \}.
\end{equation*}
In particular, the unbounded part of the pseudospectrum escapes towards $+\infty$ as $\eps$ decreases.

Additionally, we give two examples of non-selfadjoint Schr\"odinger operators outside of our class and study their pseudospectra in more detail. 
\end{abstract}
\begin{small}
\end{small}

\newpage

%\tableofcontents
%
%\newpage

\section{Introduction}\label{intro}

\subsection{Non-Selfadjoint Operators and Pseudospectra}

Let $\mathcal H$ be a Hilbert space and $A:\mathcal H\supset \mathcal{D}(A)\to\mathcal H$ be a closed operator. The \emph{spectrum} of $A$ is defined as
\begin{equation*}
 	\sigma(A) := \left\{ z\in\mathbb C\,|\, z-A \text{ is not bijective } \right\}.
\end{equation*}
If $A$ is a selfadjoint operator, the spectrum of $A$ contains a large amount of information about $A$, such as 
\begin{itemize}
\setlength{\itemsep}{0.5mm}
	\item[-]
	Does $A$ generate a one-parameter semigroup? 
	\item[-]
	Large $t$-behaviour of $\|e^{-tA}\|$,
	\item[-]
	Norm of the resolvent $\|(z-A)^{-1}\|$ for arbitrary $z\in\rho(A)$,
	\item[-]
	Location of $\sigma(A+V)$ if $V$ is a bounded perturbation.
\end{itemize}
In addition, if $A$ has compact resolvent, the eigenvectors of $A$ form a basis. 

For \emph{non-selfadjoint} operators, however, \emph{none} of the above properties can, in general, be deduced from the spectrum. This demonstrates that for non-selfadjoint operators the spectrum contains very little information about $A$. The following example provides an informative illustration of this fact. For $c\in\R$ consider the non-normal differential operator 
\begin{equation}\label{mainop}
 	H_{c}=-\frac{\mathrm{d}^2}{\mathrm{d}x^2}+ix^3+cx^2
\end{equation}	
on its maximal domain $\operatorname{dom}(H_{c})=\{\phi\in L^2(\mathbb
R):H_{c}\phi \in L^2(\mathbb R)\}$. The spectrum of $H_{c}$ is shown
in Figure \ref{specfig}.

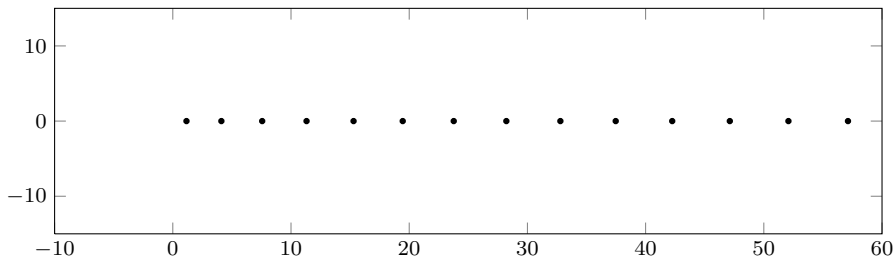
\begin{figure}[htbp]
\scriptsize
\begin{center}
\input{specfig.tex}
\end{center}
\caption{The spectrum of $H_{c}$ for $c=1$, obtained in MATLAB using the EigTool package and a modified code from \cite{Trefethen200102,TE}.}
\label{specfig}
\end{figure}

It was shown in \cite{DDT} that the spectrum of $H_{c}$ is real and positive. Moreover, $H_{c}$ is closed and has compact resolvent \cite{CGM,Mezincescu2001} so the spectrum is also discrete. On the other hand, Novak recently obtained the following result
\begin{theorem}[\cite{NV}]\label{novak}
	The operator $H_{c}$ has the following properties:
		\begin{enumerate}[(i)]
		\item
		The eigenfunctions of $H_{c}$ do not form a (Schauder) basis in $\LR$. 
		\item
		$-iH_{c}$ does not generate a bounded semigroup. 
		\item
		$H_{c}$ is not similar to a self-adjoint operator via bounded and boundedly invertible transformations.
		\end{enumerate}
 \end{theorem} 
This theorem makes it clear that $H_{c}$ is very different from a selfadjoint operator even though its spectrum looks well-behaved. 

The  above considerations motivate the definition of a \emph{finer} indicator than the spectrum for non-selfadjoint operators.
\begin{de}
	For any closed operator $A$ and $\eps>0$ the set
	\begin{equation*}
		\sigma_\eps(A) := \sigma(A)\cup\left\{ z\in\rho(A) : \|(z-A)^{-1}\|>\tfrac 1 \eps \right\}
	\end{equation*}
	is called the $\eps$-\emph{pseudospectrum} of $A$.
\end{de}
In contrast to the spectrum, the pseudospectrum \emph{does} contain significant information about $A$, such as \cite{TE,SK}
\begin{itemize}[-]
\setlength{\itemsep}{0.5mm}
	\item
	whether the eigenvectors of $A$ form a Riesz basis,
	\item
	whether $A$ generates a one-parameter semigroup,
	\item
	the large-$t$ behaviour of $e^{-tA}$.
\end{itemize}
In addition to the above, one has the following characterization of the pseudospectrum
\begin{equation*}
 	\sigma_\eps(A) = \bigcup_{\|V\|<\eps} \sigma(A+V)	
\end{equation*}
showing that the pseudospectrum contains information about the stability of the spectrum of $A$ under bounded perturbations with small norm.
For $c=0$ it was shown by Krej\v{c}i\v{r}\'{i}k and Siegl \cite{SK2} that the pseudospectrum of $H_c$ always contains an unbounded component. More precisely, they showed that for every $\delta>0$ there exist constants $C_1,C_2>0$ such that for all $\eps>0$
\hypersetup{citecolor=black}
\begin{equation}\label{pessim}
	\sigma_\eps( H_{0} )\supset\left\{ z\in\mathbb C :  |z|\geq C_1, \,|\arg z|<\left(\f\pi 2-\delta\right),\,|z|\geq C_2\left(\log\f 1 \eps\right)^{6/5} \right\}.
\end{equation}
This shows that the large eigenvalues of $H_0$ are highly unstable under small perturbations. A similar result for $c=1$ was shown by Nov\'{a}k in \cite{NV} and is easily extended to arbitrary $c>0$. Figure \ref{numplot} shows a numerical computation of the pseudospectrum of $H_1$. 

\newlength\figurewidth
\newlength\figureheight
\setlength\figurewidth{11cm}
\setlength\figureheight{7cm}
\begin{figure}[htbp]
	\begin{center}
		\includegraphics[scale=0.7]{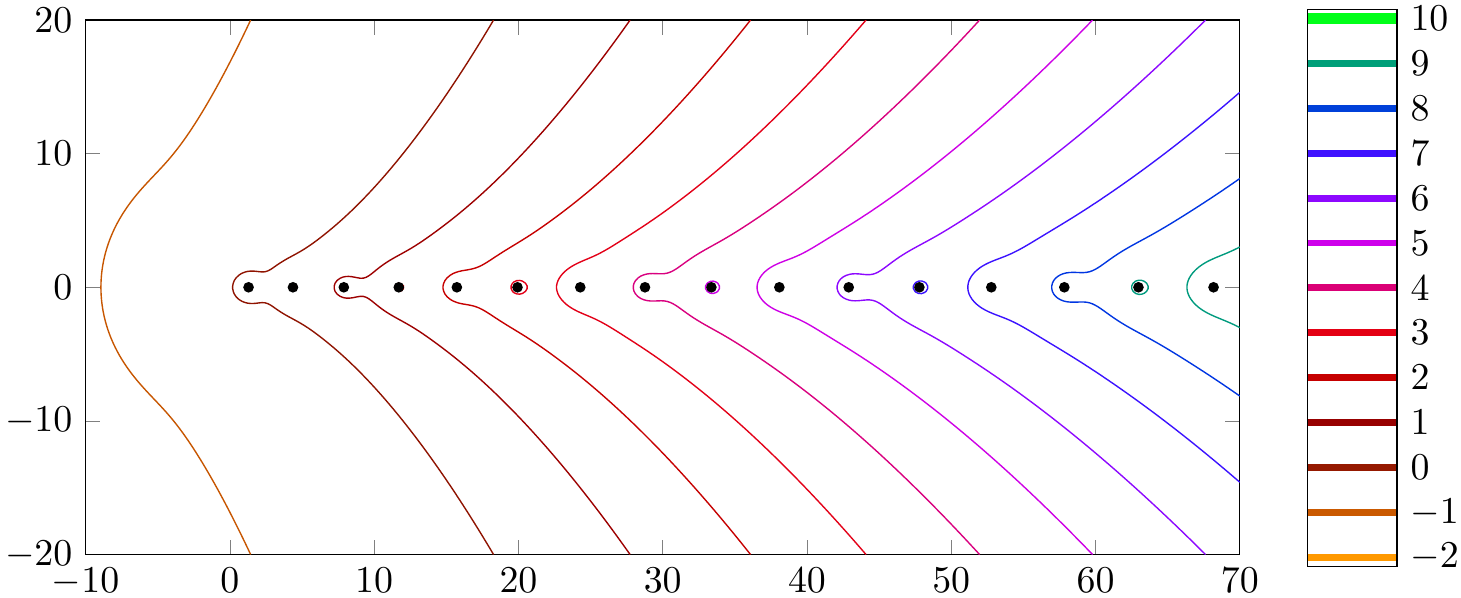}
	\end{center}
	\caption{Numerical plot of the lines of constant resolvent norm of $H_1$ also obtained using the EigTool package and a modified code from \cite{Trefethen200102,TE}. The colour bar shows the values of $\log_{10}(\|(\lambda-H_1)^{-1}\|)$.}\la{numplot}
\end{figure}

Equation \eqref{pessim} and Figure \ref{numplot} make it clear that for every fixed $\eps$ the pseudospectrum of $H_c$ contains a whole sector in the complex plane for $c>0$. Moreover, the opening angle of the sector can be chosen arbitrarily close to $\pi$ provided that a ball of sufficiently large radius around 0 is removed. In particular, large eigenvalues are very unstable under small perturbations.

On the other hand, Figure \ref{numplot} suggests that the unbounded component of the pseudospectrum escapes towards $+\infty$ as $\eps\to 0$. All of this suggests that the \emph{lower} eigenvalues of $H_c$ should indeed be stable (for $c>0$) under small perturbations of $H_c$, despite the above results.

It should be noted that the operator $H_c$ was first considered in the works of Bender et al. (see e.g. \cite{Bender1998,Bender:1998gh,Bender:2007nj}) who studied it in the context of Quantum Mechanics.
\\

In this paper we will study a class of non-normal Schr\"odinger operators containing the operators $H_c,\; (c>0)$. More precisely, we will prove a bound on the pseudospectrum of the operator $H=-\Delta+V$, where $\re V(x)\geq c|x|^2-d$  for some $c,d>0$ on $L^2(\R^n)$, which complements the results of \cite{SK2,NV}. 

The next section will contain a precise definition of the operator of interest and state our main results. Section \ref{examples} contains two illustrative special cases outside of our class, in which more precise bounds can be obtained. Finally, we will discuss several open problems in Section \ref{conclusion}.

\section{The Operator of Interest and Main Results}

Unless otherwise stated, the notation $L^2(\R^n)$ will always denote $L^2(\R^n,\mathbb C)$. The same convention holds for other function spaces. Motivated by the examples in the introduction, we are going to investigate Schr\"odinger Operators with growing real parts. 

\subsection{Definition of the Operator}

To begin with, let us quote results by \cite{BST2015} and \cite{EE} which allow the rigorous definition of a large class of Schr\"odinger operators.\footnote{
	The original proposition in \cite{BST2015} in fact allows even more general potentials than the one we state here.
}
\begin{prop}[\cite{BST2015,EE}]\label{BSTprop}
	Let $V\in  W^{1,\infty}_{\text{loc}}(\R^n)$ be a function such that
	\begin{enumerate}[(i)]
	\setlength\itemsep{0.5mm}
	\item $\re V \geq 0$
	\item 
	There exist $a,b>0$ such that $|\nabla V|^2\leq a + b|V|^2$
	\item $V$ is unbounded at infinity:	$|V(x)|\to\infty\quad\text{as}\quad |x|\to\infty$
	\end{enumerate}
	Then we have the following.
	\begin{enumerate}
	\setlength\itemsep{0.5mm}
	\item
	The minimal operator 
	\begin{equation}
		H_{\text{min}}:=-\Delta+V,\quad \mathcal D(H_{\text{min}}):=C^\infty_0(\R^n)
	\end{equation}
	is closable on $L^2(\R^n)$ with closure
	\begin{equation*}
		T=-\Delta+V,\quad\mathcal D(T)=H^2(\R^n)\cap \{\psi\in L^2(\R^n) : Vf\in L^2(\R^n)\};
	\end{equation*}
	\item
	$T$ is m-accretive;
	\item
	The resolvent of $T$ is compact.
	\end{enumerate}
\end{prop}
Using the above proposition, let us define an operator $H$ on $L^2(\R^n)$ as follows. 
\begin{de}\label{Hdef}
	Let $V:\R^n\to\mathbb C$ satisfy the conditions of Prop \ref{BSTprop} and assume in addition that there exist constants $c,d>0$ such that
	\begin{equation}
		\re V(x)\geq c|x|^2-d.
	\end{equation}
	We denote by $H$ the linear operator $H:\mathcal D(H)\to L^2(\R^n)$ as the closure of
	$$H_{\text{min}}:=-\Delta+V\quad\text{ on }\quad C_0^\infty(\R^n).$$
	according to Proposition \ref{BSTprop}.
\end{de}

\subsection{Main Results}

From now on, unless otherwise stated, $H$ will denote the operator defined in Definition \ref{Hdef}. Our first result is the following.
\begin{lemma}\label{compactness}
	The one-parameter semigroup generated by $-H$ is immediately compact.
\end{lemma}
This is used to prove our main theorem
\begin{theorem}\label{mainth}
Let $H$ be defined as in Definition \ref{Hdef}. Then for every $\delta,R>0$ there exists an $\eps>0$ such that
	\begin{equation}\label{maininc}
 		\sigma_\eps(H)\subset \{z:\re z \geq R\}\cup\bigcup_{\lambda\in\sigma(H)}\{z:|z-\lambda|<\delta \}.
	\end{equation}
\end{theorem}
We immediately obtain the following corollary about the so-called harmonic oscillator with imaginary cubic potential. 
\begin{cor}
Let $$H_c = -\f{d^2}{dx^2} + ix^3 + cx^2$$ for some $c>0$ be defined on $\mathcal D(H_c)=H^2(\R)\cap\{\psi\in\LR : x^3\psi\in
\LR\}\subset\LR$. Then one has the inclusion \eqref{maininc} for the pseudospectrum of $H_c$.
\end{cor}
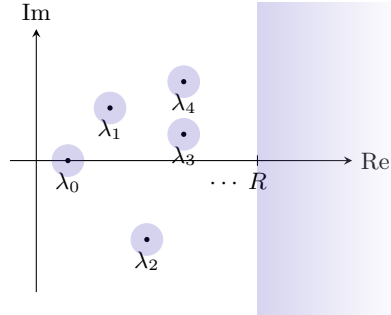
\begin{wrapfigure}[12]{r}{0.4\textwidth} 
\setlength{\abovecaptionskip}{3pt plus 1pt minus 1pt}
\setlength{\belowcaptionskip}{3pt plus 1pt minus 2pt}
%\begin{figure}
\vspace{-5.5mm}
\centering
\begin{tikzpicture}[>=stealth,scale=0.7]
\footnotesize{
    \path[use as bounding box] (-0.5,-3) rectangle (7,3);
    \definecolor{niceblue}{rgb}{0.2,0.15,0.7}
    \draw [->] (0,-2.5)--(0,2.5) node[above]{Im};
    \draw [->] (-0.5,0)--(6,0) node[right]{Re};
    \draw (0.6,0) node[circle,fill,inner sep=.7pt,label=below:$\lambda_0$]{} ;
    \draw (1.4,1) node[circle,fill,inner sep=.7pt,label=below:$\lambda_1$]{} ;
    \draw (2.1,-1.5) node[circle,fill,inner sep=.7pt,label=below:$\lambda_2$]{} ;
    \draw (2.8,0.5) node[circle,fill,inner sep=.7pt,label=below:$\lambda_3$]{} ;
    \draw (2.8,1.5) node[circle,fill,inner sep=.7pt,label=below:$\lambda_4$]{} ;
    \draw (3.6,0) node[label=below:$\cdots$]{} ;
%    \draw (4.6,-1.5) node[circle,fill,inner sep=.7pt,label=below:$\lambda_m$]{} ;
    \draw (4.2,-.1) --(4.2,.1) node[label=below:$R$]{} ;
    \filldraw[niceblue,opacity=.2] (2.1,-1.5) circle (.3cm) ;
    \filldraw[niceblue,opacity=.2] (2.8,0.5) circle (.3cm) ;
    \filldraw[niceblue,opacity=.2] (2.8,1.5) circle (.3cm) ;
    \filldraw[niceblue,opacity=.2] (1.4,1)  circle (.3cm) ;
    \filldraw[niceblue,opacity=.2] (.6,0) circle (.3cm) ;
    \shade[left color=niceblue,opacity=.2,right color=white] (4.2,-3) rectangle (7,3);}
%    \draw[color=white] (9,0)--(9.01,0);
 \end{tikzpicture}
    \caption{The pseudospectrum of $H$ is contained in sets of the above shape.}
%\end{figure}
\vspace{-5mm}
\end{wrapfigure}
We remark that the inclusion \eqref{maininc} is optimal in the sense that the unbounded component of the pseudospectrum cannot be contained in a sector of opening angle less than $\pi$ as  \eqref{pessim} shows. 

Moreover, Theorem \ref{mainth} can be seen as complementary to the results of \cite{NV}. Indeed, while it was shown there that there always exist infinitely many eigenvalues which are highly unstable under bounded perturbations, our result shows that the \emph{lower} eigenvalues (that is, those with small real part) \emph{do} remain stable if the perturbation is small enough in norm.

The method of proof of Theorem \ref{mainth} is inspired by \cite{Boulton} and based on estimates of the semigroup generated by $-H$.

\section{Proof of Theorem \ref{mainth}}

In this section we will first prove Lemma \ref{compactness} and then use it to prove Theorem \ref{mainth}. Throughout this section, $H$ denotes the operator defined in Lemma \ref{compactness} and we will make frequent use of properties 1., 2., 3. of Proposition \ref{BSTprop} without further reference.

\subsection{Proof of Lemma \ref{compactness}}

 It is well-known \cite{DW,Brezis} that for all $\phi_0\in L^2(\R^n)$ the semigroup generated by $-H$ is nothing but the solution operator to the initial value problem

\begin{equation}\la{complexeq}
 	\begin{cases}
 		\partial_t\phi &= -H\phi \\
	 	\phi(0) &= \phi_0.
 	\end{cases}
\end{equation}
In this section we will show that the operator $e^{-tH}$ is compact on $ L^2(\R^n) $. The first step will be to turn \eqref{complexeq} into a coupled system of real equations and then using the results of \cite{DL}.

\paragraph{Writing the equation as a system.}

We will use the fact that $L^2(\R^n,\mathbb C)$ is canonically isomorphic to $L^2(\R^n,\R^2)$. In the following we will denote this isomorphism by $U:L^2(\R^n,\mathbb C)\to L^2(\R^n,\R^2)$.

Now, let us write $\phi(x)=f_1(x)+i f_2(x)$. A straightforward calculation shows that \eqref{complexeq} is equivalent to the system
\begin{align}
\begin{dcases}\label{system}
 	\del_t f_1 &= \Delta f_1+\imag(V) f_2-\re(V)f_1 \\
 	\del_t f_2 &= \Delta f_2-\imag(V)f_1-\re(V)f_2
 \end{dcases}
\end{align}
which we will write as
\begin{align*}
 	\del_t \begin{pmatrix}
 		f_1 \\ f_2
 	\end{pmatrix}
 	&=
 	\left[
 		\Delta + Q(x)
 	\right]
 	\begin{pmatrix}
 		f_1 \\ f_2
 	\end{pmatrix} \\
 	&=:-\hat H\begin{pmatrix}
 		f_1 \\ f_2
 	\end{pmatrix},
\end{align*}
where $Q(x)=\begin{pmatrix} -\re V(x) & \imag V(x) \\ -\imag V(x) & -\re V(x) \end{pmatrix}$. Along the lines of 
\cite{DL}
 we define $\kappa(x):=-c|x|^2+d$ (with $c,d$ from Definition \ref{Hdef}) which satisfies the estimate
\begin{equation}
	\langle Q(x)\xi,\xi\rangle\leq\kappa(x)\|\xi\|^2\qquad\forall\xi\in\R^2,
\end{equation}
according to our assumptions about $V$. We also define the scalar differential operator\footnote{
	More precisely, $\hat H_{2\kappa}$ should be regarded as the $L^2$-closure of the operator initially defined on the space $\mathcal S(\R^n,\R^2)$.
}
\begin{equation}
 	\hat H_{2\kappa}:=\Delta+2\kappa(x)  \quad\text{ on }\quad L^2(\R^n,\R).
\end{equation}
The operators $-\hat H$ and $H_{2\kappa}$ satisfy Hypothesis 2.1 of 
\cite{DL} enabling us to prove the following lemma by following the lines of the proof of \cite[Prop. 2.4]{DL}.

\begin{lemma}\label{U<SF}
	Let $ f^0\in \mathcal{S}(\R^n,\R^2)$. There exists a unique classical solution to the initial value problem $[$\eqref{system}, $f(0,\cdot)=f^0$$]$ and one has
	\begin{equation}
		|f(t,\cdot)|^2 \leq  e^{t\hat{H}_{2\kappa}} \big(| f^0 |^2\big),\quad t\geq 0.
	\end{equation}
\end{lemma}
\begin{proof}
	This proof uses the local H\"older continuity of $V$. By \cite[Th. 2.6]{DL} there exists a unique classical solution $f=(f_1,f_2)$ for our choice of initial condition. Let us now multiply the first equation of \eqref{system} by $f_1$ and the second by $f_2$ and add the resulting equations. We obtain
	\begin{equation*}
		\f 1 2 \del_t |f|^2 = f\cdot \Delta f - \re(V)|f|^2.
	\end{equation*}
	Using the product rule this may be rewritten as 
	\begin{align*}
		\del_t |f|^2 &= (\Delta - 2\re V)|f|^2 -2|\nabla f|^2 \\
		&= \big( \Delta - 2\kappa(x)  - 2W(x)\big)|f|^2 -2|\nabla f|^2 \\
		&= \hat{H}_{2\kappa}  \big(|f|^2\big) -2\big(W(x)|f|^2 + |\nabla f|^2\big),
	\end{align*}
	where we have defined $W(x):=\re V(x) - \kappa(x)  \geq 0$.
	Now, define $w:=|f|^2- e^{t\hat{H}_{2\kappa}} \big(|f^0|^2\big)$. We obviously have $w(0,\cdot)=0$ and from the above calculation we obtain
	\begin{equation*}
		(\del_t - \Delta + 2\kappa(x)  )w\leq 0,\quad t>0.
	\end{equation*}
	Thus applying the maximum principle \cite[Prop. 2.3 (ii)]{DL} we obtain $w\leq 0$.
\end{proof}

\paragraph{The operator $\boldsymbol{\hat H_{2\kappa}}$.}

Regarded as an operator on $L^2(\R^n,\R)$, the operator { $-\hat H_{2\kappa}$  is of course nothing but the harmonic oscillator with frequency $\omega=\sqrt{8c}$, shifted by the constant $-2d$. Its negative is well-known to generate a one-parameter semigroup $e^{t\hat H_{2\kappa}}$  which can be represented by the Mehler kernel
\begin{align*}
	\big( e^{t\hat H_{2\kappa}}  g\big)(t,x) &=  e^{2td}\Big(\f{2\pi}{\omega} \sinh(2\omega t)\Big)^{-\f 1 2} \int  
	e^{
	-\f{\omega}{2}\f{\cosh(2\omega t)(|x|^2+|y|^2)-2x\cdot y}{\sinh(2\omega t)}
	}
	g(y)\,dy\\
	&=: \int K(t,x,y)g(y)\,dy
\end{align*}
(cf. \cite[Chapter 7.2]{DA}).

\begin{lemma}\label{K<Gauss}
	Let $t>0$ and $0<\alpha \leq \cosh(2\omega t)-1$ and define $\mu(x):=e^{-\f{\alpha\omega}{2\sinh(2\omega t)} |x|^2}$. Then 
	\begin{equation}\label{K<C}
		|K(t,x,y)| \leq C\, \mu(x)\mu(y),
	\end{equation}
	where $C$ depends only on $t$ and $\omega$.
\end{lemma}
\begin{proof}
	We only have to check that $-\alpha(|x|^2+|y|^2)\geq -\cosh(2\omega t)(|x|^2+|y|^2)-2x\cdot y$. This follows immediately from the assumption on $\alpha$. Note that $\cosh(2\omega t)-1>0$ for $t>0$, so such an $\alpha$ exists.
\end{proof}
Note that this lemma implies that  $e^{t\hat H_{2\kappa}}$  is a Hilbert-Schmidt operator.

\paragraph{Compactness of $\boldsymbol{e^{-tH}}$.
}
The following lemma states that a cut-off version of $e^{-tH}$ converges in norm to $e^{-tH}$.

\begin{lemma}\label{RntoT}
	Let $t>0$ and $\theta_n\in C_c(\R^n)$ such that $\chi_{B_{r_n}(0)}\leq \theta_n\leq \chi_{B_{2r_n}(0)}$, where $r_n$ is defined such that
\begin{equation}
 	\sup_{x\in \R^n\backslash B_{r_n}(0)} \big(\mu(x)\big) < \f{ 1}{ n^2}
\end{equation}
(where $\mu$ was defined in Lemma \ref{K<Gauss}) and define the operator $R_n(t)$ by $$R_n(t) f := (Ue^{-\f t 2 H}U^{-1})\big(\theta_n (Ue^{-\f t 2  H}U^{-1}) f\big).$$ Then
\begin{equation}
 		\|Ue^{-tH }U^{-1}-R_n(t)\|_{L^2(\R^n,\R^2)}\to 0\qquad (n\to\infty).
\end{equation}
\end{lemma}

\begin{proof}
Let $n\in\mathbb N$ and $ f\in \mathcal{S}(\R^n,\R^2)$ and compute
\begin{align*}
 	|Ue^{-tH}U^{-1} f(x)-R_n(t) f(x)|^2 \!&\stackrel{\scriptscriptstyle{{\text{Lemma }\ref{U<SF}}}}{\leq} \!\!
 	 e^{t\hat H_{2\kappa}}  \big(|Ue^{-\f t 2  H }U^{-1} f-\theta_n (Ue^{-\f t 2  H }U^{-1}) f|^2\big)(x) \\
 	&\stackrel{\phantom{\scriptscriptstyle{\text{Lemma }\ref{U<SF}}}}{=}\!\!\!\!\!\!\int K\big(\tfrac t 2 ,x,y\big)\big|(1-\theta_n(y))(Ue^{-\f t 2  H }U^{-1}) f(y)\big|^2\,dy
\end{align*}
Now integrate both sides over $x$.
\begin{align*}
 	\|Ue^{-tH }U^{-1} f-R_n(t) f\|_{L^2}^2 &\leq \iint K\big(\tfrac t 2 ,x,y\big)\big|(1-\theta_n(y))(Ue^{-\f t 2  H }U^{-1}) f(y)\big|^2\,dx dy \\
 	&\leq C \iint \mu(x)\mu(y)|1-\theta_n(y)|^2\,|(Ue^{-\f t 2  H }U^{-1}) f(y)|^2\,dxdy\\
 	&\leq \!C\! \left(\!\int\!\!\mu(x)dx\!\right)\! \|\mu(y)(1\!-\theta_n(y))^2\|_{\scriptscriptstyle{\infty}}\! \!\int\! |(Ue^{-\f t 2  H }U^{-1}\!) f\!(y)|^2 dy \\
 	&\leq C' \Big(\sup_{\scriptscriptstyle{y\in \R^n\backslash B_{r_n}}}\! \mu(y)\Big)\, \|(Ue^{-\f t 2  H }U^{-1}) f\|_{L^2}^2 \\
 	&\leq \f{ M}{ n^2} \|(Ue^{-\f t 2  H }U^{-1}) f\|_{L^2}^2
\end{align*}
for some $M>0$. Using the unitarity of $U$ and the fact that $e^{-\f t 2 H }$ is a bounded operator on $L^2(\R^n,\mathbb C)$ we finally arrive at
\begin{equation}
 	\|Ue^{-tH }U^{-1} f-R_n(t) f\|_{L^2(\R^n,\R^2)}^2 \leq \left(\f{ M}{ n^2} \|e^{-\f t 2  H }\|^2\right)\| f\|_{L^2(\R^n,\R^2)}^2.
\end{equation}
By density of $\mathcal S(\R^n,\R^2)$ we conclude that this inequality is valid for all $ f\in L^2(\R^n,\R^2)$. This immediately yields
\begin{equation}
 	\|Ue^{-tH }U^{-1}-R_n(t)\|_{L^2(\R^n,\R^2)} \leq \f {L}{ {n}}
\end{equation}
for some $L>0$.
\end{proof}
We can now use Lemma \ref{RntoT} to prove Lemma \ref{compactness}.
By closedness of the set of compact operators and Lemma \ref{RntoT} we only have to show that $R_n(\tau)$ is compact for every $n$ (cf. Theorem 7.1.4 in \cite{DA2}). Since furthermore $Ue^{-\f \tau 2 H }U^{-1}$ is a bounded operator on $L^2(\R^n,\mathbb C)$, we only show that $T_n(\tau) := \theta_n Ue^{-\f \tau 2  H }U^{-1}$ is compact.

This compactness can be established by a standard local regularity argument for parabolic PDE. Consider thus the function 
$$
v(x,t) := \eta(t)\xi_n(x) f(x,t),
$$ 
with $\eta(t) \in C^\infty([0,1])$ such that $\eta(0)=0$ and $\eta(t) = 1$ for $t>\tau/4$ and with $\xi_n \in C^\infty(\R^n)$, $\xi_n\ge 0$ such that $\chi_{B_{4r_n}}\leq \xi_n \leq \chi_{B_{8r_n}}$. A simple calculation shows that $v$ satisfies
$$
v_t + H v = w_n
$$
where $w_n = \xi_n(\del_t\eta)f-\eta(\Delta\xi_n)f-2\eta\nabla \xi_n\cdot \nabla f$. Note that $w_n$ is in the dual space of $H^1$ with bounded norm $\|w_n\|_{L^2(0,1;H^{-1}(\R^n))} \le C \|f\|_{L^2(0,1;L^2(\R^n))}$ for some $C>0$ as it contains only first spatial distributional derivatives of an $L^2$ function. By a standard Galerkin approximation we therefore immediately obtain an $L^2(0,1;H^1(\R^n))$-bound on $v$, which implies the same bound on $u$ on the domain $(\tau/4, 1)\times B_{4r_n}(0)$ where $f$ and $v$ agree.

Repeating this procedure, but cutting off between $\tau/4$ and
$\tau/2$ in time and $2r_n$ and $4r_n$ in space, yields higher
regularity of $f$ on $(\tau/2,1)\times B_{2r_n}(0)$ by noting that
the new $w_n$ is now bounded in $L^2(0,1;L^2(\R^n))$. The bound we
obtain is thus 
$$
\|f\|_{L^\infty(\tau/2,1;H^1(B_{2r_n}))} \le C \|f\|_{L^2(0,1;L^2(\R^n))} \le C\|f^0\|_{L^2(\R^n)},
$$
Since $\theta_n$ is smooth and $\operatorname{supp} \theta_n \subset \overline{B_{2r_n}(0)}$, this implies that
$$
\|T_n(\tau)f^0\|_{H^1(\R^n)} \le C\|f^0\|_{L^2(\R^n)}.
$$
The desired compactness result follows immediately.
\qed
\begin{cor}\label{norcon}
	The semigroup $e^{-tH}$ is norm-continuous. 
\end{cor}
\begin{proof}
	This follows from the above compactness result, together with compactness of $(H-\lambda)^{-1}$; cf. \cite[Th. 4.29]{engel2000one}
\end{proof}
Note that by Corollary \ref{norcon}, we obtain the following bound on the spectrum of $H$ (cf. \cite[Th. 4.18]{engel2000one}):
\begin{cor}\label{specbound}
	Let $b\in\R$. Then the set
	$$ 
	\left\{ \lambda\in\sigma(H) : \re\lambda\le b \right\}
	$$
	is bounded.
\end{cor}

\subsection{Bound on the Pseudospectrum}

Let us first quote several theorems from the theory of one-parameter semigroups of operators which will be needed in the sequel. Throughout this section, $T(t)$ denotes a (strongly continuous) one-parameter semigroup and $-A$ its generator.

\begin{theorem}[{\cite[Ch. 2.2]{DA}}]\la{compact}
Let $T(t)$ be a one-parameter semigroup and assume $T(a)$ is compact for some $a>0$. Then, the spectrum of $A$ consists of a countable discrete set of eigenvalues each of finite multiplicity, and 
\begin{equation}\label{SMT}
 	\sigma(T(t))=\{0\}\cup \{ e^{-t\lambda}\,|\,\lambda\in\sigma(A) \}.
\end{equation}
\end{theorem}

\begin{theorem}[{\cite[Ch. 1.2]{DA}}]\la{radius}
	If $T(t)$ is a one-parameter semigroup on a Banach space then
	\begin{equation}
 		a:=\lim_{t\to\infty} t^{-1}\log\|T(t)\|
	\end{equation}
	exists with $-\infty\leq a<\infty$. Moreover
	\begin{equation}
 		r(T(t)):=\max\{|\lambda|:\lambda\in\sigma(T(t))\}=e^{at}\quad\forall t>0,
	\end{equation}
\end{theorem}

\begin{theorem}[{\cite[Ch. VII.4]{DW}}]\la{werner}
	Let $T(t)$ be a one-parameter semigroup with $\|T(t)\|\leq Me^{at}$ for all $t\geq 0$. Then
	\begin{equation}
 		\|(z-A)^{-1}\| \leq \f{M}{a-\re z} \qquad\forall z:\re z < a.
	\end{equation}
\end{theorem}

\paragraph{Example: The imaginary Airy Operator.}

The theorems in the previous section can be used to estimate the pseudospectra of m-accretive operators. As an illustrative example, let us treat the imaginary Airy operator defined as
\begin{equation}
 	H_{Ai}=-\f{d^2}{dx^2}+ix\quad\text{ on }\quad \{\phi\in L^2(\R)\,|\,-\phi''+ix\phi\in L^2(\R)\}.
\end{equation}
This operator is m-accretive, and thus generates a one-parameter semigroup. Using the Fourier transform one can show that 
\cite{DA2}

\begin{equation}
 	\|e^{-tH_{Ai}}\|=e^{-\f{t^3}{12}}.
\end{equation}
Let now $a>0$. Choosing $M_a = \sup_{t\geq 0} \big(e^{at-\f{t^3}{12}}\big)$ we have
$$ e^{-\f{t^3}{12}}\leq M_ae^{-at}$$
and so
\begin{equation}
 	\|e^{-tH_{Ai}}\|\leq M_ae^{-at}.
\end{equation}
Thus Theorem \ref{werner} tells us that 
\begin{equation}
 	\|(z-H_{Ai})^{-1}\| \leq \f{M_a}{a-\re z}\qquad\forall z:\re z< a.
\end{equation}
(note that the generator of the semigroup is not $H_{Ai}$ but $-H_{Ai}$). In particular, we have for (say) $\re z<a-1$ that
\begin{equation}\label{Haiestimate}
 	\|(z-H_{Ai})^{-1}\| \leq M_a.
\end{equation}
This shows that for $\varepsilon<\f{1}{M_a}$ the set $\{z\,|\,\re z<a-1\}$ does not intersect the $\varepsilon$-pseudospectrum. In more suggestive terms: The $\varepsilon$-pseudospectrum wanders off to the right as we decrease $\varepsilon$.

A simple calculation shows that $M_a=\sup_{t\geq 0} \big(e^{at-\f{t^3}{12}}\big)=e^{\f{4}{3}a^{3/2}}$. This even enables us to estimate how fast the pseudospectrum moves with decreasing $\eps$. To this end, let $z\in\sigma_\eps(H_{Ai})$ for some fixed $\eps>0$. Then by \eqref{Haiestimate} we have 
\begin{align*}
 	\f 1 \eps &\leq \|(z-H_{Ai})^{-1}\| \\
 	&\leq  e^{\f{4}{3}(\re z+1)^{3/2}}  \\
 	&\leq e^{w(\re z)^{3/2}} 
\end{align*}
for some $w>0$ and $\re z$ large enough. This inequality immediately leads to 
\begin{equation}\label{scaling}
 	\re z \geq w^{-1} \left( \log\f 1 \eps\right)^{2/3},
\end{equation}
with $w$ independent of $\eps$. This shows that indeed every point in the $\eps$-pseudospectrum moves towards $+\infty$ with velocity at least $\left( \log\f 1 \eps\right)^{2/3}$.

Let us compare this to the results of \cite{SK}. Using semiclassical techniques the authors showed that there exist constants $C_1,C_2>0$ such that for all $\eps>0$
\begin{equation*}
 	\sigma_\eps(H_{Ai})\supset \left\{ z : \re(z)\geq C_1,\,\re(z)\geq C_2\Big(\log\f 1\eps\Big)^{2/3} \right\}.
\end{equation*}
Equation \eqref{scaling} confirms that the scaling found in \cite{SK} is in fact optimal.
 The same result has previously been obtained in \cite{Willi} using a different method of proof.
 
 Note that together with the observation that $\|(H_{Ai}-z)^{-1}\|$ is independent of $\operatorname{Im}(z)$ (see \cite[Problem 9.1.10]{DA2}) the pseudospectrum of $H_{Ai}$ is (essentially) completely characterized: it consists of half-planes moving towards $+\infty$ with asymptotic velocity $\big(\log\f 1\eps\big)^{2/3}$.

\paragraph{The General Case: A First Estimate.}

Let us now turn back to the operator $H =-\Delta+V $ of Definition \ref{Hdef}. To conclude the proof of Theorem \ref{mainth} we will need several lemmas which will be established next. By Theorem \ref{compact} we know that
\begin{equation}\label{spm}
 	\sigma(e^{-tH })=\{0\}\cup \{ e^{-t\lambda}\,|\,\lambda\in\sigma(H ) \}.
\end{equation}
%Let us assume w.l.o.g. that $\inf\{\re\lambda : \lambda\in\sigma(H)\}>0$. If this is not the case, 
%let $k$ be the number of eigenvalues with zero real part. By Corollary \ref{specbound} we know that $k<\infty$. Now, we can project out the eigenspaces associated with these $k$ eigenvalues. This only results in a finite-dimensional error which is easily controlled (see Lemma \ref{BoultonLemma} below). Let $\lambda_0$ denote an eigenvalue of $H$ with $\re\lambda_0=\min\{\re\lambda : \lambda\in\sigma(H),\re\lambda>0\}$. 
Let us denote the eigenvalues of $H$ by $\lambda_j$ such that $\re\lambda_j \leq \re\lambda_i$ for $j \leq i$. Thus, $\lambda_0$ denotes an eigenvalue with minimal real part. In fact, up to now we could have $\re\lambda_0=-d$. We will account for this problem below in Lemma \ref{BoultonLemma}. With this notation, we obtain from eq. \eqref{spm} that
\begin{equation}
 	r(e^{-tH })=e^{-t\re\lambda_0},
\end{equation}
Thus by Theorem \ref{radius} we have 
\begin{equation}
 	-\re\lambda_0=\lim_{t\to\infty}t^{-1}\log\|e^{-tH }\|.
\end{equation}
In other words, we have that for every $\alpha<\re\lambda_0$ that
\begin{equation}
 	\lim_{t\to\infty}e^{\alpha t}\|e^{-tH }\|=0.
\end{equation}
Let such an $\alpha<\re\lambda_0$ be fixed and choose $t_\alpha$ such that $e^{\alpha t}\|e^{-tH }\|<1$ for all $t>t_\alpha$. On the whole we have
\begin{alignat*}{2}
 	\|e^{-tH }\|&<e^{-\alpha t} &&\quad \forall t>t_\alpha\\
 	\|e^{-tH }\|&\leq 1 && \quad\forall t>0 \quad\text{(since it is a contraction semigroup)},
\end{alignat*}
so we finally arrive at 
\begin{equation}
 	\|e^{-tH }\|\leq M_\alpha e^{-\alpha t}\qquad\forall t>0,
\end{equation}
with $M_\alpha=e^{\alpha t_\alpha}$.

We are now in the position to proceed as for the imaginary Airy operator. Theorem \ref{werner} tells us that 
\begin{equation}\label{resab}
 	\|(z-H )^{-1}\| \leq \f{M_\alpha}{\alpha-\re z}\qquad\forall z:\re z< \alpha.
\end{equation}
Note, however, that this time we cannot simply let $\alpha\to +\infty$ since we are restricted to $\alpha<\re\lambda_0$.

\paragraph{Pushing the Pseudospectrum Towards Infinity.}
Let $Q_n=\f{1}{2\pi i}\oint_\gamma(H -z)^{-1} dz$ denote the Riesz projection associated with $H $, where $\gamma$ encloses only the $n$-th eigenvalue $\lambda_n$ (which is possible since the spectrum of $H $ is discrete). Moreover, define $P_m:=\sum_{n=0}^m Q_n$. Then each of the operators $Q_n,P_m$ commutes with the resolvent of $H$.

Since $H$ has compact resolvent, we have that $\dim(\ran Q_n)<\infty\quad\forall n$. For each $m\in\mathbb N$ the space $ L^2(\R^n)$ decomposes into a direct sum of closed, $H $-invariant subspaces\footnote{
	$H $-invariance follows from the fact that the $Q_n$ commute with $H $ and closedness of $\Ran(I-P_m)$ follows from the Fredholm alternative.
}
\begin{equation}\la{directsum}
	 L^2(\R^n)=\Ran Q_0\oplus\cdots\oplus \Ran Q_m\oplus\Ran(I-P_m)
\end{equation}
Because $e^{-tH}$ commutes with the resolvent of $H$, each of the above subspaces is invariant under $e^{-tH }$ and hence, by \cite[Sect. II.2.3]{engel2000one} the generator of $e^{-tH }\vert_{\Ran Q_n}$ is $-H \vert_{\Ran Q_n}$. The same is true for $\Ran(I-P_m)$. 

Since the spectrum of $H |_{\Ran(I-P_m)}$ is $\{\lambda_n : n>m\}$ (and since the restriction of a compact operator is compact), applying Theorem \ref{compact} again gives
\begin{equation}\label{verschsp}
 	\sigma\big(e^{-tH }\big\vert_{\Ran(I-P_m)}\big)=\{0\}\cup \{e^{-t\lambda_n}\}_{n=m+1}^\infty.
\end{equation}

\begin{lemma}\label{BoultonLemma}
For all $z\in\rho(H )$, one has
\begin{equation}\label{normineq}
 	\|(H -z)^{-1}\|\leq C\left( \sum_{n=0}^m\|(H |_{\Ran Q_n}-z)^{-1}\| + \| (H |_{\Ran(I-P_m)}-z)^{-1} \| \right)
\end{equation}
where $C$ depends only on $\|Q_n\|$ ($n\leq m$).
\end{lemma}
\begin{proof}
	Let $z\in\rho(H )$ and $\xi,\psi\in L^2(\R^n)$ such that $(H -z)\xi=\psi$ and $\|\psi\|=1$. We want to estimate $\|\xi\|$. To do this, note that by surjectivity of $(H -z)$ we have
	\begin{equation}\label{L2H}
 		L^2(\mathbb R^n) =  \left(\bigoplus_{n=0}^m  \Ran(H  |_{\Ran( Q_n)} -z)\right)  +  \Ran(H  |_{\Ran (I-P_m)} -z).
	\end{equation}
	Note that the first term on the right hand side is actually equal to $\bigoplus_{n=0}^m\Ran Q_n$, since $\ran Q_n$ is $H$-invariant.
	\begin{itemize}
	\item[\emph{Claim:}] We have $\Ran(I-P_m)= \Ran(H  |_{\Ran(I-P_m)} -z)$.
	\item[\emph{Proof of Claim:}] Since the $Q_n$ commute with $H $, we have 
	\begin{align*}
		\Ran(H  |_{\Ran(I-P_m)} -z)&=\Ran\big((H  |_{\Ran(I-P_m)} -z)(I-P_m)\big)\\
		&=\Ran\big((I-P_m)(H  |_{\Ran(I-P_m)} -z)\big)\\
		&\subset  \Ran(I-P_m).
	\end{align*}
	Now, suppose there was a $0\neq\phi\in\Ran(I-P_m)\backslash \Ran(H  |_{\Ran(I-P_m)} -z)$. Since \eqref{directsum} is a direct sum $\phi$ cannot have any components in $\bigoplus_{n=0}^m\Ran Q_n$. But then $\phi\notin \Ran(H -z)$, by \eqref{L2H}, which contradicts surjectivity.
	\end{itemize}	
	Now, decompose 
	\begin{align*}
		\psi&=\sum_{n=1}^mQ_n\psi+(I-P_m)\psi \\
		&=:\sum_{n=1}^m\psi_n+\tilde\psi.
	\end{align*}
	  Choose $\xi_n\in\Ran Q_n$ such that $(H -z)\xi_n=\psi_n$ and $\tilde\xi\in\Ran(I-P_m)$ such that $(H -z)\tilde\xi=\tilde\psi$ (which is possible since $\Ran(I-P_m)= \Ran(H  |_{\Ran(I-P_m)} -z)$). But now it is clear that
	  \begin{alignat*}{2}
		\|\xi_n\| &\leq \|(H |_{\Ran Q_n}-z)^{-1}\|\|\psi_n\| &&\leq \|(H |_{\Ran Q_n}-z)^{-1}\|\|Q_n\|\|\psi\| \\
		\|\tilde\xi\| &\leq \|(H |_{\Ran(I-P_m)}-z)^{-1}\|\|\tilde\psi\| &&\leq \|(H |_{\Ran(I-P_m)}-z)^{-1}\|\|(I-P_m)\|\|\psi\| 
	\end{alignat*}
	Finally, using the triangle inequality we obtain
	\begin{align*}
 		\|\xi\|& \leq \sum_{n=1}^m\|\xi_n\| + \|\tilde\xi\| \\
 		&\leq     \left( \sum_{n=0}^m\|Q_n\|\|(H |_{\Ran Q_n}-z)^{-1}\|+ \|(I-P_m)\|\|(H |_{\Ran(I-P_m)}-z)^{-1}\|\right)\!\|\psi\| \\
 		&\leq \bigg(1+\sum_{n=0}^m\|Q_n\|\bigg)\Big(\|(H |_{\Ran Q_n}-z)^{-1}\|+\|(H |_{\Ran(I-P_m)}-z)^{-1}\|   \Big)
	\end{align*}
	which concludes the proof.
\end{proof}
We are finally able to complete the proof of Theorem \ref{mainth}. In \eqref{normineq} the first term on the right hand side is nothing but a sum of the resolvents of matrices. These are well-known to decay in norm at infinity. In fact, a simple calculation shows that one has $\|(T-\lambda)^{-1}\|\leq \big(|\lambda|-\|T\|\big)^{-1}$ as $|\lambda|\to\infty$. As a consequence, the $\eps$-pseudospectra of $(H|_{\Ran Q_n}-z)^{-1}$ are contained in discs around the $\lambda_n$ for $\eps$ small enough.

For the second term we can use \eqref{verschsp} in Theorems \ref{radius} and \ref{werner} to obtain an estimate similar to \eqref{resab}, but with $\alpha<\re\lambda_{m+1}$ instead. By Corollary \ref{specbound} we necessarily have $\re\lambda_n\to\infty$ as $n\to\infty$. Thus we obtain a bound on $\|(H-\lambda)^{-1}\|$ on vertical lines with arbitrarily large real part and the proof of Theorem \ref{mainth} is completed.

\section{The cases of vanishing and negative real part of the potential}\label{examples}

It is natural to ask whether the condition $\re V(x)\ge c|x|^2{\color{red}{-d}}$ can be relaxed. In this section we will discuss two examples giving hints as to what might or might not be possible. First, we will consider an example of a Schr\"odinger operator with $\re V=0$ which still satisfies the inclusion \eqref{maininc}. Second, we will show that in the case $\re V(x)\le -c|x|^2$ one can not expect any inclusion of the form \eqref{maininc}.

\subsection{Example: The Imaginary Cubic Oscillator}

In this section we consider the operator 
\begin{equation}
 	H_B=-\f{d^2}{dx^2} + ix^3 \quad\text{ on }\quad \LR,
\end{equation}
defined in the sense of Proposition \ref{BSTprop}. $H_B$ is sometimes called the imaginary cubic oscillator, or the Bender oscillator. We immediately obtain closedness of $H_B$, compactness of its resolvent and m-accrevity from Proposition \ref{BSTprop}. Moreover, it is known \cite{DDT,Shin} that the spectrum of $H_B$ is entirely real and positive which enables us to number the eigenvalues $\lambda_i$ of $H_B$ such that $\lambda_i\leq\lambda_j$ for $i\le j$ and $\lambda_0>0$.
In this section, we will prove the following result about $H_B$. 

\begin{theorem}\label{cubicmainth}
	For the pseudospectrum of $H_B$ the inclusion \eqref{maininc} holds and in
addition there exists a $C>0$ such that for every $\delta>0$ there is
an $\eps>0$ such that 
	\begin{equation}\label{mainscaling}
		\sigma_\eps(H_B)\subset  \left\{z:\re z \geq C\left(\log\f{1}{\eps}\right)^{6/5}\right\}\cup\bigcup_{\lambda\in\sigma(H_B)}\{z:|z-\lambda|<\delta \}.
	\end{equation}
	In particular, apart from disks around the eigenvalues, the $\eps$-pseudospectrum is contained in the half plane $\left\{\re z \geq C\left(\log\f{1}{\eps}\right)^{6/5}\right\}$.
\end{theorem}
\begin{proof}
As in the previous section we want to estimate $\|e^{-tH_B}|_{\Ran(I-P_m)}\|$ for $m\in\mathbb N$. We know that the eigenfunctions of $H_B$ form a complete set in $\LR$ and the algebraic eigenspaces are one-dimensional \cite{SK2,DucTai2006351}.
Thus, we can use Lemma 3.1 of \cite{Davies}:
\begin{lemma}[\cite{Davies}]
	Let $T(t)$ be a strongly continuous semigroup and $\{\psi_n\}_{n=1}^\infty$  a complete set of linearly independent vectors. Let $T_n(t)$ denote the restriction of $T(t)$ to $\mathrm{span}\{\psi_1,\dots,\psi_n\}$. Then
	\begin{equation}
		\|T(t)\|=\lim_{n\to\infty}\|T_n(t)\|
	\end{equation}
	for all $t\geq 0$.
\end{lemma}
From now on, let $\{\psi_n\}_{n=1}^\infty$ denote the set of eigenvectors of $H_B$ and let $V_m^n$ $:=$ $\operatorname{span}\{\psi_m,$ $\dots,\psi_n\}$ = $\bigoplus_{k=m}^{n}\Ran(Q_k)$. The Lemma now implies
\begin{align*}
 	\|e^{-tH_B}|_{\Ran(I-P_{m-1})}\| = \lim_{n\to\infty}\|e^{-tH_B}|_{V_m^n}\|.
\end{align*}
The analytic functional calculus (see \cite[Ch.V.]{taylor1980introduction}) shows that $\sum_{k=m}^nQ_k$ is a projection again and thus we have $\psi=\sum_{i=m}^{n}Q_i\psi$ for every $\psi\in V_m^n$ which we can use as follows.
\begin{align*}
 	\|e^{-tH_B}|_{\Ran(I-P_m)}\psi\| &= \lim_{n\to\infty}\|e^{-tH_B}|_{V_m^n}\psi\| \\
 	&= \lim_{n\to\infty}\left\|\sum_{k=m}^{n}e^{-t\lambda_k}Q_k\psi\right\| \\
 	&\leq \lim_{n\to\infty}\sum_{k=m}^{n}e^{-t\lambda_k} \| Q_k \| \| \psi \| \\
 	&= \left(\,\sum_{k=m}^{\infty}e^{-t\lambda_k}\|Q_k\| \right) \|\psi \| 
\end{align*}
so we obtain
\begin{equation}
 	 	\|e^{-tH_B}|_{\Ran(I-P_m)}\| \leq \sum_{k=m}^{\infty}e^{-t\lambda_k}\|Q_k\|.
\end{equation}
In \cite{Henry} it was shown that $\lim_{k\to\infty}\f{\log \|Q_k\|}{k}=\f{\pi}{\sqrt 3}$. Accordingly, for every $\mu>\f{\pi}{\sqrt 3}$ there exists a $C>0$ such that 
\begin{equation}
 	\|Q_k\|\leq Ce^{\mu k}.
\end{equation}
In particular, choosing $\mu=2$, we obtain $\|Q_k\|\leq Ce^{2 k}$ for some $C>0$.

On the other hand, it is well-known from \cite{sibuya1975global} that 
\begin{equation}
 	\lambda_k \geq ck^{6/5}.
\end{equation}
Combining these two facts, we arrive at
\begin{align*}
 	 	\|e^{-tH_B}|_{\Ran(I-P_m)}\| &\leq \sum_{k=m}^{\infty}e^{-tck^{6/5}}Ce^{2k} \\
 	 	&= C\sum_{k=m}^{\infty}e^{-tck^{6/5}+2k} 
\end{align*}
Clearly, there exists a $k_0$ such that $\f 1 2 tck^{6/5}>2k$ for all $k>k_0$ and $k_0$ is independent of $t$ as long as (say) $t\geq 1$. So we can decompose
\begin{align*}
 	\|e^{-tH_B}|_{\Ran(I-P_{m-1})}\| &\leq C\sum_{k=m}^{k_0}e^{-tck^{6/5}+2k} + C\!\sum_{k=k_0+1}^{\infty}e^{-\f{c }{ 2} tk^{6/5}}
\end{align*}
Since $k_0$ is independent of $m$ and $t$, the first term in this estimate is only present as long as $m<k_0$. 

Since we are interested in asymptotics, let us assume $m>k_0\geq 1$ from now on. Our task is thus to estimate the second term in the above inequality. This is easily done by using $\lfloor x+1 \rfloor\geq x$ for all $x>0$ and calculating
\begin{align*}
 	\sum_{k=m}^\infty e^{-\f c 2 t (k+1)^{6/5}} &\leq \int_m^\infty e^{-\f c 2 t x^{6/5}}dx \\
 	&\leq \int_m^\infty \big(\tfrac 6 5 x^{1/5}\big) e^{-\f c 2 t x^{6/5}}dx \\
 	&= \tfrac{ 2}{ ct} \left[ -e^{-\tfrac c 2 t x^{6/5}}\right]_{m}^{\infty} \\
 	&= \tfrac{ 2}{ ct} e^{-\tfrac c 2 t m^{6/5}} 
\end{align*}
This finally shows our main ingredient
\begin{lemma}
	There exist constants $k_0,M,\omega>0$ such that 
	\begin{equation*}
		\|e^{-tH_B}|_{\Ran(I-P_{m-1})}\| \leq Me^{-\omega m^{\f 6 5}t}
	\end{equation*}
	for all $m>k_0,\,t\geq 1$.
\end{lemma}
\noindent This immediately leads to\footnote{
	Since we only know that $\|e^{-tH_B}\|$ is bounded by 1 between $t=0$ and $t=1$, we might need to increase $M$ to obtain \eqref{resest}.} 
\begin{equation}\label{resest}
 	\|(H_B|_{\Ran(I-P_{m-1})}-z)^{-1}\| \leq \f{\tilde M}{\omega m^\f{6}{5}-\re z}
\end{equation}
for all $\re z<\omega m^\f{6}{5}$, where $\tilde M,\omega $ are independent of $m$. On the whole, the resolvent of $H_B$ is estimated by (see the proof of Lemma \ref{BoultonLemma})
\begin{align*}
 	\big\|\!(H_B\!-\!z)\!^{-1}\!\big\| \!
 	&\leq \!\bigg(\!1\!+\!\sum_{k=1}^m\!\|Q_k\|\!\bigg)\!\bigg(\sum_{k=1}^m\!\big\|\!(H_B|_{\Ran Q_k}\!-\!z)\!^{-1}\!\big\|\!+\!\big\|\!(H_B|_{\Ran(I-P_m)}\!-\!z)\!^{-1}\!\big\| \!  \bigg) \\
 	&\leq \!\bigg(1+\sum_{k=1}^{m}\|Q_k\|\bigg)\!\bigg( \sum_{k=1}^m\f{1}{|\lambda_k-z|} + \f{\tilde M}{\omega (m+1)^\f{6}{5}-\re z} \bigg)
\end{align*}
The first summand in the second factor gives the discs around the eigenvalues in \eqref{maininc}, the second gives the half-plane. 
%Since we are interested in the $\eps$-dependence of the \emph{unbounded} component, let us examine the second summand $\left(1+\sum_{k=1}^{m}\|Q_k\|\right) \|(H_B|_{\Ran(I-P_m)}-z)^{-1}\| $ more closely.
If we keep the distance of $\mathrm{Re}(z)$ to $\omega (m+1)^{6/5}$ constant, the second factor on the right-hand side stays bounded as $m\to\infty$. 
Since the first factor grows as $e^{(\text{constant})\cdot m}$, we have 
\begin{equation}\label{sim}
 	\|(H_B-z)^{-1}\| \leq C e^{C'(\re z)^{5/6}}
\end{equation}
uniformly in $z$ as long as $\operatorname{dist}(z,\sigma(H_B))$ is bounded below by a positive constant.

Keeping this in mind, suppose now that $z\in\sigma_\eps(H_B)\cap \{\operatorname{dist}(z,\sigma(H_B))>1\}$. We deduce 
\begin{align*}
 	\log\left(\f 1 \eps\right) \leq \log \|(H_B-z)^{-1}\| &\leq C''(\re z)^{5/6} \\
 	\Leftrightarrow\; \left(\log\f 1 \eps\right)^{6/5} &\leq C''\,\re z
\end{align*}
Together with the complementary estimate in \eqref{pessim}   this proves the scaling in \eqref{mainscaling}.
\end{proof}

Let us compare Theorem \ref{cubicmainth} to the results of \cite{SK2}}
. As noted in the introduction, it was shown there that for every $\delta>0$ there exist constants $C_1,C_2>0$ such that for all $\eps>0$
\begin{equation*}
	\sigma_\eps( H_B )\supset\left\{ z\in\mathbb C :  |z|\geq C_1, \,|\arg z|<\left(\f\pi 2-\delta\right),\,|z|\geq C_2\left(\log\f 1 \eps\right)^{6/5} \right\}.
\end{equation*}
Clearly, we have found the same scaling in \eqref{mainscaling}. Thus, Theorem \ref{cubicmainth} shows that the scaling \eqref{pessim} obtained in \cite{SK} is sharp.

Moreover, we obtain as a byproduct the following two statements about the semigroup and the resolvent of $H_B$.
\begin{cor}\label{cor1}
	The semigroup $e^{-tH_B}$ is immediately differentiable.
\end{cor}
\begin{cor}\label{cor2}
	The resolvent norm of $H_B$ satisfies
	\begin{equation}
		\lim_{r\to\infty} \|(H_B-s-ir)^{-1}\|=0
	\end{equation}
	for all $s\in\R$.
\end{cor}
\begin{proof}
By \cite[Cor. II.4.15]{engel2000one} and the estimate \eqref{resest} the semigroups $e^{-tH_B}|_{\Ran(I-P_{m})}$ are immediately differentiable for every $m$ and hence immediately norm-continuous. By \cite[Cor. II.4.19]{engel2000one} one has
\begin{equation*}
 	\lim_{r\to\infty}\big\|\big(H_B|_{\Ran(I-P_{m-1})}-(s+ir)\big)^{-1}\big\| \to 0 \quad \forall s<\omega m^{\f 6 5}.
\end{equation*}
Together with the estimate \eqref{normineq} the assertion follows.
\end{proof}

Notice that the strategy of the proof of Theorem \ref{cubicmainth} also applies to more general operators for which the norms of the spectral projections are known such as those considered in \cite{Henry2,MSV}

\subsection{Counterexample: An Operator with Negative Real Part}

Let $c > 0$ and consider the operator
\begin{equation}
	H_c = -\f{d^2}{dx^2} + ix^3 - cx^2 
\end{equation}

on $L^2(\R)$, where $H_c$ is defined in the sense of Proposition \ref{BSTprop}. This operator is still well-behaved in the sense that it is closed and its resolvent is compact. Moreover, its spectrum is still well-behaved in the sense that it is closed and its resolvent is compact. Moreover, its spectrum is still real and positive. However, as we will show, its pseudospectrum is not well-behaved at all. In fact, $H_c$ does not even generate a one-parameter semigroup in this case.

\begin{theorem}\label{negreal}
	For $H_c$ no inclusion of the type \eqref{maininc} is possible. More precisely, for every $C,R,M>0$ there exists $z\in\mathbb C$ such that $\re z<-R,\,|z|>M$ and
	\begin{equation}
		\|(H_c-z)^{-1}\| \geq C.
	\end{equation}
	In particular, $H_c$ does not generate a one-parameter semigroup.
\end{theorem}
\begin{wrapfigure}{R}{0.37\textwidth} 
\vspace{-5mm}
\centering
\input{SemiClassicalTIKZ.tex}
\caption{The semiclassical pseudospectrum of $H_c^h$. The boundary curve approaches the imaginary axis as $h\to 0$.}
\label{semiclpssp}
\vspace{-1cm}
\end{wrapfigure}
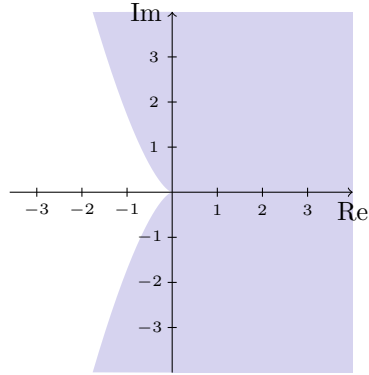
\begin{proof}
We will use Theorem 3.1 and Lemma 4.1 of \cite{NV}. Similarly to their strategy, let us define the unitary transformation
\begin{equation*}
 	(U\psi)(x):=\tau^{1/2}\psi(\tau x),
\end{equation*}
where $\tau>0$. This transformation takes $H_c$ to its semiclassical analogue 
\begin{equation*}
 	H_c^h:=\tau^{-3}UH_cU=-h^2\f{d^2}{dx^2}+ix^3-ch^{2/5}x^2,
\end{equation*}
where $h=\tau^{-5/2}$.

The \emph{semiclassical pseudospectrum} (cf. (3.2) in \cite{NV}) for this operator is the set (cf. Figure \ref{semiclpssp})
\begin{equation*}
 	\Lambda_h=\{\xi^2+ix^3-ch^{2/5}x^2 : \xi,x\neq 0\}.
\end{equation*}
We obviously have $i\in\Lambda_h$ for every $h>0$ (remember that $c<0$). By \cite[Theorem 3.1]{NV} and the unitarity of $U$ there exists a $C>0$ such that 
\begin{align*}
	\|(H_c-i\tau^3)^{-1}\| &= \tau^{-3}\|(H_c^h-i)^{-1}\| \\
	&\geq h^{6/5}C^{1/h}
\end{align*}
Sending $\tau=h^{-2/5}\to\infty$, we see that the resolvent norm of $H_c$ diverges exponentially on the imaginary axis.

To show divergence on vertical lines with strictly negative real part we may shift $H_c$ by a real constant and then apply the above procedure. More precisely, let $\alpha>0$ and consider the operator $H_c+\alpha$. Its semiclassical analogue is
\begin{equation*}
 	\tau^{-3}U(H_c+\alpha)U=H_c^h+h^{6/5}\alpha
\end{equation*}
and its semiclassical pseudospectrum 
\begin{equation*}
 	\Lambda_h=\{\xi^2+ix^3-ch^{2/5}x^2+h^{6/5}\alpha : \xi,x\neq 0\}
\end{equation*}
is shifted to the right by $h^{6/5}\alpha$. Its boundary curve intersects the imaginary axis when $-ch^{2/5}x^2+h^{6/5}\alpha=0$ the solution of which is $h^{2/5}\left(\f{\alpha}{c}\right)^{1/2}$. Since this tends to 0 as $h\to 0$ one can always find $h_0>0$ such that $i\in\Lambda_h$ for all $h<h_0$. This enables us to apply the above procedure for the shifted operator and obtain again exponential divergence on the imaginary axis.
\end{proof}
\paragraph{Remark:}
Given the above lower estimate of $\|(H_c-z)^{-1}\|$, let us mention that it is still possible to obtain weaker upper bounds on the resolvent norm of $H_c$. Boegli, Siegl and Tretter have shown in \cite{BST2015} that for a very general class of Schroedinger operators, including  $H,H_c$ and $H_B$, the resolvent norm always decays in a sector in the complex plane which opens to the \emph{left}. 

In other words, operators such as $H_c$ are still \emph{sectorial} in the sense of \cite{haase2006functional}. In particular, there exists an analytic functional calculus for these operators which, in turn, yields the existence e.g. of fractional powers of $H_c$.

\section{Conclusion and Outlook}\label{conclusion}

We have shown that for $\re V\ge c|x|^2$ the unbounded component of the pseudospectrum of $H=-\Delta+V$ moves towards $+\infty$ as $\eps\to 0$. We note that this result holds for arbitrary imaginary part of the potential.

For a similar operator with $\re V=0$ we were able to give a precise scaling for how fast this happens. To obtain this scaling the knowledge of the norms of the Riesz projections was crucial. 

Let us remark that an analogous result to Theorem \ref{mainth} trivially holds for operators which are m-sectorial (in the sense of \cite{Kato}). This is due to the fact that the resolvent norm decays outside the numerical range. This includes e.g. the Bender oscillator $-\f{d^2}{dx^2}-(ix)^\nu,\;2<\nu<4$ (cf. \cite{Mezincescu2001} for a precise definition). The conclusion of Theorem \ref{mainth} holds for $H$ if $2<\eps\leq 3$. Furthermore, by semiclassical methods, the conclusion of Theorem \ref{negreal} holds if $3<\eps<4$.

More generally, Schr\"odinger Operators with a potential whose range is contained in a sector belong to the above category (cf. \cite[Prop. 2.2]{BST2015} for a precise study).

A number of open questions remain.
\begin{itemize}
\item[-]
To the authors' knowledge the norms of the Riesz projections of the harmonic oscillator with imaginary cubic potential have not been computed yet, but we strongly suspect that the scaling $\|Q_k\|\sim e^{\omega k}$ (which holds for the Bender oscillator) is also true in this case. 
\item[-]
Furthermore, we have seen that the resolvent norm of the Bender oscillator $H_B$ goes to zero on vertical lines in the complex plane. However, we do not know the rate of the decay. Clearly, there exists no $C>0$ such that 
$$\|(H_c-s-ir)^{-1}\|\leq \f{C}{|r|},\quad\forall s\in\R$$
because this would imply that $H_B$ generates an \emph{analytic} semigroup (which is false by \eqref{pessim}).
The question remains exactly how slow the decay is. The answer could be used to confirm the results of \cite{Willi} who computed the asymptotic shape of the level sets of the resolvent norm.
\item[-]
Finally, there is the obvious question as to whether the central assumption $\re V\ge c|x|^2$ can be relaxed. It is not obvious how to generalize our method of proof to potentials which do not satisfy this lower bound. Indeed, our compactness proof of the semigroup heavily relied on the fundamental solution of the harmonic oscillator. However, the examples of the imaginary cubic oscillator and the imaginary airy operator suggest that the lower bound on $\re V$ is not essential. It seems likely to the authors that under suitable conditions on  $\imag V$ the semigroup of $-\Delta+V$ will be compact even for $\re V=0$.
\end{itemize}

\section*{Acknowledgements}
PWD and FR gratefully acknowledge inspiring discussions with P.~Siegl
(Bern), D.~Krej\v{c}i\v{r}\'{i}k (Prague), R.~Nov\'{a}k (Prague), and
L.~Boulton (Edinburgh) as well as support from the European Commission
via the Marie Sk{\l}odowska-Curie Career Integration Grant \emph{STP}.
PED was supported in part by an STFC Consolidated Grant, ST/L000407/1,
and also by the GATIS Marie Curie FP7 network (gatis.desy.eu) under
REA Grant Agreement No 317089.

\section*{References}
\begingroup
\renewcommand{\section}[2]{}
\nocite{*}
\bibliography{Paper-Allg.bib}
\bibliographystyle{alphaabbr}
\endgroup

\end{document}

%% file: specfig.tex
\begin{tikzpicture}
\begin{axis}[%
width=11cm,
height=3cm,
at={(0,0)},
scale only axis,
point meta min=-10,
point meta max=2,
unbounded coords=jump,
xmin=-10,
xmax=60,
ymin=-15,
ymax=15,
axis background/.style={fill=white}
]
\addplot [color=black,line width=1.0pt,mark size=0.7pt,only marks,mark=*,mark options={solid},forget plot]
  table[row sep=crcr]{%
1.15626707198445	0\\
4.1092287528238	0\\
7.56227385497055	0\\
11.3144218198317	0\\
15.2915537535813	0\\
19.451529107096	0\\
23.7667405659721	0\\
32.7890832204992	0\\
28.2175246126817	0\\
57.1200123315513	0\\
62.2373782346685	0\\
37.4698260934543	0\\
42.2503888400521	0\\
52.0811584253371	0\\
47.123204144387	0\\
};
\end{axis}
\end{tikzpicture}%

%% file: SemiClassicalTIKZ.tex
\begin{tikzpicture}[domain=0:1.587,scale=0.6]
\definecolor{niceblue}{rgb}{0.2,0.15,0.7}
\fill[fill=niceblue,opacity=.2] plot ({-0.7*\x^2},{\x^3}) -- (0,4) ;
\fill[fill=niceblue,opacity=.2] plot ({-0.7*\x^2},{-\x^3}) -- (0,-4) ;
\fill[fill=niceblue,opacity=.2] (0,-4) rectangle (4,4);
\draw[-] (1,-.1) -- (1,.1) ;
\draw (1,-0.4) node{\tiny{$1$}};
\draw[-] (2,-.1) -- (2,.1) ;
\draw (2,-0.4) node{\tiny{$2$}};
\draw[-] (3,-.1) -- (3,.1) ;
\draw (3,-0.4) node{\tiny{$3$}};
\draw[-] (-1,-.1) -- (-1,.1) ;
\draw (-1,-0.4) node{\tiny{$-1$}};
\draw[-] (-2,-.1) -- (-2,.1) ;
\draw (-2,-0.4) node{\tiny{$-2$}};
\draw[-] (-3,-.1) -- (-3,.1) ;
\draw (-3,-0.4) node{\tiny{$-3$}};
\draw[-] (-.1,-1) -- (.1,-1) ;
\draw (-0.5,-1) node{\tiny{$-1$}};
\draw[-] (-.1,-2) -- (.1,-2) ;
\draw (-0.5,-2) node{\tiny{$-2$}};
\draw[-] (-.1,-3) -- (.1,-3) ;
\draw (-0.5,-3) node{\tiny{$-3$}};
\draw[-] (-.1,1) -- (.1,1) ;
\draw (-0.4,1) node{\tiny{$1$}};
\draw[-] (-.1,2) -- (.1,2) ;
\draw (-0.4,2) node{\tiny{$2$}};
\draw[-] (-.1,3) -- (.1,3) ;
\draw (-0.4,3) node{\tiny{$3$}};
\draw[->] (-3.6,0) -- (4,0) node[below]{\small{Re}};
\draw[->] (0,-4) -- (0,4) node[left]{\small{Im}};
\end{tikzpicture}